\documentclass[10pt,twocolumn,twoside]{IEEEtran}
\IEEEoverridecommandlockouts

\usepackage{enumitem}
\usepackage{url}
\usepackage{cite}
\usepackage{amsmath,amssymb,amsfonts}
\usepackage{algorithmic}
\usepackage{graphicx}
\usepackage{textcomp}
\usepackage{balance}
\usepackage{xcolor}
\usepackage{bm}
\usepackage{booktabs}   
\usepackage{multirow}
\usepackage{subcaption}

\usepackage{bbm}
\usepackage[utf8]{inputenc}
\usepackage{array}
\usepackage{amssymb, amsthm, amsfonts}
\newtheorem{proposition}{Proposition}
\newtheorem{lemma}{Lemma}
\newtheorem{corollary}{Corollary}
\newtheorem{theorem}{Theorem}
\def\BibTeX{{\rm B\kern-.05em{\sc i\kern-.025em b}\kern-.08em
    T\kern-.1667em\lower.7ex\hbox{E}\kern-.125emX}}
\usepackage{lipsum}

\newcounter{subassumption}[assumption]

\title{\LARGE \bf 
  Optimal Solar Investment and Operation under Asymmetric Net Metering
}

\author{Nathan Engelman Lado$^\dagger$, Ahmed S. Alahmed$^\ddagger$, Audun Botterud$^\dagger$, Saurabh Amin$^\dagger$
\vspace{-0.5cm}
 \thanks{ 
$^\dagger$LIDS, MIT, Cambridge, MA, USA. Email: {\tt engelm25@mit.edu, audunb@mit.edu, amins@mit.edu}.}
\thanks{$^\ddagger$Electrical Engineering, King Fahd University of Petroleum and Minerals, Dhahran, Saudi Arabia. Email: {\tt alahmad@kfupm.edu.sa}.}
}

\def\beq{\begin{equation}}
\def\eeq{\end{equation}}
\def\bea{\begin{eqnarray}}
\def\eea{\end{eqnarray}}
\def\ba{\begin{array}}
\def\ea{\end{array}}

\def\bitem{\begin{itemize}}
\def\eitem{\end{itemize}}
\def\ben{\begin{enumerate}}
\def\een{\end{enumerate}}

\def\ie{{\it i.e.,\ \/}}

\definecolor{bgrd}{rgb}{1,1,1}
\definecolor{gray}{rgb}{0.5,0.5,0.5}
\definecolor{dkr}{rgb}{0.7,0.1,0.2}
\definecolor{dkb}{rgb}{0.1,0.1,0.8}

\begin{document}
\begingroup
\allowdisplaybreaks

\maketitle

\begin{abstract}
We examine the joint investment and operational decisions of a \emph{prosumer}, a customer who both consumes and generates electricity, under net energy metering (NEM) tariffs. Traditional NEM schemes provide temporally flat compensation at the retail price for net energy exports over a billing period. However, ongoing reforms in several U.S. states are introducing time-varying prices and asymmetric import/export compensation to better align incentives with grid costs.
While prior studies treat PV capacity as exogenous and focus primarily on consumption behavior, this work endogenizes PV investment and derives the marginal value of solar capacity for a flexible prosumer under asymmetric NEM tariffs. 
We characterize optimal investment and show how optimal investment changes with prices and PV costs. Through this analysis, we identify a PV effect: changes in NEM pricing in one period can influence net demand and consumption in generating periods with unchanged prices through adjustments in optimal PV investment. The PV effect weakens the ability of higher import prices to increase prosumer payments, with direct implications for NEM reform. We validate our theoretical results in a case study using simulated household and tariff data derived from historical conditions in Massachusetts.
\end{abstract}



\section{Introduction}
The rapid deployment of behind-the-meter distributed energy resources (DERs), particularly rooftop photovoltaic (PV) systems and flexible loads, is transforming electricity consumption and generation patterns at the distribution level. Customers equipped with such resources, commonly termed prosumers, are able to both import electricity from and export surplus generation to the distribution utility (DU)~\cite{Burger&Perez&Arriaga:19CEEPR}. Net energy metering (NEM) has emerged as the predominant tariff mechanism for facilitating bidirectional energy exchange in the United States. Under conventional NEM, a customer's net energy consumption is calculated over a billing period. Net imports are charged at the prevailing retail price and net exports are compensated with credits, typically valued at the retail rate (symmetric NEM) or at a reduced price (asymmetric NEM)~\cite{NEMevolution:23NAS}.%
\footnote{Symmetric and asymmetric NEM designs are also referred to in the literature as full NEM and partial NEM, respectively.} 
Although many existing NEM tariffs remain temporally time-invariant and symmetric, a growing number of jurisdictions are transitioning toward time-varying and asymmetric import/export tariff structures. These reforms aim to better reflect marginal system costs, improve grid efficiency, and align prosumer incentives with distribution system needs~\cite{Next10Report2022,Alahmed&Tong:22ACMSEIR}. Notable examples of asymmetric NEM tariffs include those recently adopted in California, Arizona, Georgia, Michigan, and Illinois~\cite{DSIRE}.

Much of the existing literature on NEM examines either operational DER decisions under fixed capacity or DER adoption in aggregate. Early work studies prosumer behavior, benefits, and PV adoption under legacy NEM tariffs with symmetric prices \cite{Borenstein:15NBER, Darghouth&Barbose&Wiser:11EP,Sun&Feng&Tong:20TAC, Yamamoto:17EP}. Within this literature, operational analyses often treat PV capacity as exogenous and focus on how prosumers schedule consumption and DER operation given a fixed level of installed capacity.

More recent studies consider operational decisions and DER aggregation under time-varying and asymmetric NEM tariffs \cite{Alahmed&Tong:22TSG, Chakraborty&Poolla&Varaiya:19TSG}, but generally maintain the assumption of exogenous PV capacity. Other work endogenizes PV capacity under NEM without incorporating demand flexibility \cite{Comello&Reichelstein:17RSER, Chakraborty&Khargonekar:23CSL, Cerino&Noussan:18JCP}. In contrast, our work endogenizes PV capacity, incorporates prosumer flexibility, and analyzes asymmetric, time-varying NEM tariffs. We show that jointly optimizing investment and operational decisions exposes effects that are not captured when PV capacity is treated as fixed. \vspace{3pt}

\noindent\textbf{\textit{Summary of Contributions:}}
 This paper studies the joint investment and operational decisions of a prosumer with behind-the-meter PV generation and flexible loads under general NEM tariffs. More specifically:
\begin{itemize}[leftmargin=*, labelsep=0.5em]
    \item We demonstrate that optimal PV capacity depends on a weighted average of prices and marginal utilities derived from the import, export, and net-zero (self-consumption) regimes. We then derive comparative statics with respect to changes in import/export prices and solar cost.
    \item We identify the PV effect: when NEM prices change in one period, adjustments to optimal solar capacity affect prosumer consumption, net demand, and payments in all other generation periods.
    \item We show that net-zero self-consumption periods dampen the sensitivity of PV investment to price and cost changes, smoothing capacity responses relative to regimes valued at market prices.

    \item We present a numerical case study using simulated Massachusetts demand and renewable generation. The case study demonstrates why this work is relevant to regulators and utilities by illustrating how NEM tariffs affect prosumer investment, net demand, consumption, and cost recovery.
\end{itemize}

\noindent\textbf{\textit{Paper Notations and Organization:}}
A summary of notations used in the paper is provided in Table \ref{table:notations}. The rest of the paper is organized as follows. Sec.\,\ref{sec:model} presents the prosumer payment and decision model under a general NEM tariff structure. Sec.\,\ref{sec:Optimal} delineates the optimal operation and investment decisions of a prosumer under NEM. Sec.\,\ref{sec:CaseStudy} presents a case study for a prosumer under NEM using simulated data for a household in Massachusetts, followed by concluding remarks in Sec.\,\ref{sec:conclusion}. Mathematical proofs of the theoretical results and explanations of parameter choices for the case study are presented in the appendix.

\begin{table}[t]
\centering
\caption{Notation summary}
\label{table:notations}
\setlength{\tabcolsep}{6pt}
\renewcommand{\arraystretch}{1.1}
\begin{tabular}{llll}
\hline
\multicolumn{4}{l}{\textbf{\underline{Sets}}} \\[2pt]

$\mathcal{T}$ & Periods
& $\mathcal{T}^g\subseteq \mathcal{T}$ & Gen. Periods   \\
\hline
\multicolumn{4}{l}{\textbf{\underline{Parameters}}} \\[2pt]

$c^g\!\in\! \mathbb{R}_{\geq 0}$ & Cost of PV inv.& $\pi_t^+\!\in\! \mathbb{R}_{\geq 0}$ & Import price  \\
$\psi_{t}\!\in\! \mathbb{R}_{\geq 0}$ & Cap. factor
& $\pi_t^-\!\in\! \mathbb{R}$ & Export price  \\
$\bar \psi_t\!\in\! \mathbb{R}_{\geq 0}$ & Cap. factor bnd.
& $\pi^c\in \mathbb{R}$ & Fixed charge \\
$\bar g \!\in\! \mathbb{R}_{\geq 0}$ & PV cap. bnd.
& $U_{t}(\cdot)\!\in\! \mathbb{R}$ & Period-$t$ utility \\

\hline
\multicolumn{4}{l}{\textbf{\underline{Decision Variables}}} \\[2pt]

$d_{t}\!\in\!\mathbb{R}_{\geq 0}$ & Electricity cons.
& $d_{t}^+\!\in\! \mathbb{R}_{\geq 0}$ & Import\\
$g\!\in\! \mathbb{R}_{\geq 0}$ & Solar cap.
& $d_{t}^-\!\in\! \mathbb{R}_{\geq 0}$ & Export \\
\hline
\multicolumn{4}{l}{\textbf{\underline{Auxiliary Quantities}}} \\[2pt]
$P^{\mbox{\tiny NEM}}\!\in\! \mathbb{R}$ & NEM payment
& $\bar d_t^\pm\!\in\! \mathbb{R}_{\geq 0}$ & Imp./Exp. thresh. \\
$\underbar d^+\!\in\! \mathbb{R}_{\geq 0}$ & PV Threshold 
&$S\!\in\! \mathbb{R}$ & Prosumer Surplus\\
\hline
\multicolumn{4}{l}{\textbf{\underline{Dual Variables}}} \\[2pt]
$\gamma_{t}\!\in\! \mathbb{R}$ & Balance constr. dual
& $\alpha_{t}\!\in\! \mathbb{R}_{\geq 0}$ & Imp. nonneg. dual \\
$\beta_{t}\!\in\! \mathbb{R}_{\geq 0}$ & Exp. nonneg. dual
& $\kappa\!\in\! \mathbb{R}_{\geq 0}$ & PV nonneg. dual \\
$\rho\!\in\! \mathbb{R}_{\geq 0}$ & PV cap. constr. dual
&  &  \\

\hline
\end{tabular}
\end{table}
\vspace{-6pt}

\section{Prosumer Decisions Model and Payment}\label{sec:model}

\noindent\textbf{\textit{DER and Payment under NEM:}}
We consider a prosumer facing a finite set of settlement periods $t\in \mathcal{T}$. 
A prosumer consumes electricity,  $d_{t}$, and generates electricity in each period, $\psi_t g$. For the set of generating periods, $\mathcal{T}^g$, the capacity factor is a continuously distributed parameter $\psi_{t}\in [0,\bar \psi_{t}]$. For periods $t\in\mathcal{T}\setminus \mathcal{T}^g$, the capacity factor is zero, \ie $\bar\psi_{t}=0$.

 Energy balance requires that net demand equals the difference between imports and exports. When imbalanced, the prosumer imports/exports electricity to/from the grid, hence,
\vspace{-4pt}
\begin{align*}
d_{t}-\psi_t g=d_{t}^+ - d_{t}^-.
\end{align*}
\vspace{-17pt}

The prosumer's net demand, $d_{t}^+ - d_{t}^-$, in period $  t  $ is positive if consumption exceeds generation ($  d_t\! >\! \psi_t g  $), negative if generation exceeds consumption ($  d_t \!< \!\psi_t g  $), and zero if they are equal ($  d_t\! =\! \psi_t g  $). 
Using the import and export variables, the total payment under NEM is defined as follows:
\vspace{-2pt}
\begin{align}
P^{\mbox{\tiny NEM}}
:=
\pi^c
+
\sum_{t\in\mathcal T}\Big(\pi_t^+\,d_{t}^+ - \pi_t^-\,d_{t}^-\Big).
\label{eq:Pi_NEM_def}
\end{align}\vspace{-6pt}

The expected payment is
$\mathbb E[P^{\mbox{\tiny NEM}}]$, with expectation taken over $\psi_t$. We assume that the import price is greater than the export price, $\pi_t^+>\pi_t^-$, and treat the case of equal prices by subtracting a small offset from $\pi_t^-$.

\vspace{3pt}
\noindent\textbf{\textit{Prosumer Decision Problem:}}
The consumer derives utility $U_{t}(d_{t})$ from electricity consumption. We assume the utility function is increasing, twice-differentiable, and strictly concave. Following standard microeconomics literature, the prosumer's problem is modeled as a surplus maximization problem \cite{Jehle&Reny:11PrenticeHall},  where surplus is the difference between the utility of consuming electricity and incurred charges due to consumption and investment decisions:\vspace{-3pt}
\begin{subequations}\label{eqn:prosumer_prob}
\begin{align}\label{eqn:prosumer_opt}
S^* &=\hspace{-3pt} \underset{\substack{d_{t}^+,d_{t}^-\\g, d_{t}}}{\max} 
\sum_{t \in \mathcal{T}}\mathbbm{E}\Big[U_{t}(d_{t})
\hspace{-2pt}-\hspace{-2pt}\pi_t^+ d_{t}^+ \hspace{-1pt}+\hspace{-1pt} \pi_t^- d_{t}^-\Big]
\hspace{-2pt}-c^g g \hspace{-2pt}- \hspace{-2pt}\pi^c \end{align}\vspace{-12pt}\begin{align}
\text{s.t.}\quad 
&d_{t}^+-d_{t}^-= d_{t}-g\psi_{t}\ :\  \gamma_{t}&&\forall t \in \mathcal{T}\label{eqn:cons_balance}\\
& g\geq0, g\leq\bar{g} \hspace{38pt} :\kappa, \rho \label{eqn:solar_cap}\\
&d_{t}^+ \geq 0,d_{t}^-\geq 0 \hspace{25pt} :\alpha_{t},\beta_{t} &&\forall{t\in\mathcal{T}}\label{eqn:solar_nonneg}\\
& d_{t}^+\cdot d_{t}^- = 0 \quad  (\textit{Bilinear Constraint:})&&\forall t \in \mathcal{T} \label{cons:bilinear}
\end{align}
\end{subequations}
\vspace{-14pt}

The objective \eqref{eqn:prosumer_opt} includes the utility $U_{t}(d_{t})$, the payment, and the cost of investment in PV capacity.
The energy balance constraint \eqref{eqn:cons_balance} ensures that consumption equals net demand after accounting for local generation. 
The capacity limits \eqref{eqn:solar_cap} bound feasible solar investment. Constraints \eqref{eqn:solar_nonneg} impose non-negativity on imports and exports. 

The bilinear constraint \eqref{cons:bilinear} prohibits simultaneous importing and exporting. This constraint makes the problem non-convex. Fortunately, when we assume that import prices are always greater than export prices, $\pi^+_t>\pi^-_t \forall t$, constraint \eqref{cons:bilinear} is always satisfied.
This is formalized in Lemma \ref{lem:bireform} (Appendix \ref{app:bireform }), which we use to reformulate \eqref{eqn:prosumer_opt}-\eqref{eqn:solar_nonneg} without the bilinear constraint. The remaining problem is composed of a concave objective function with linear constraints and is therefore a convex optimization problem. We leverage this in the following section to derive results for the prosumer's optimal operational and investment decisions under NEM.

\vspace{-10pt}
\section{Optimal Sizing and Operational Decisions}\label{sec:Optimal}
In this section, we first describe optimal operational decisions and derive a characterization of optimal prosumer investment. We then show how operational and investment decisions respond to changes in import prices, export prices, and the marginal cost of PV. 

\vspace{3pt}
\noindent\textbf{\textit{Optimal Prosumer Operation:}}
We apply the KKT conditions to \eqref{eqn:prosumer_opt}-\eqref{eqn:solar_nonneg}, yielding the following optimality conditions:
\begin{subequations}
\vspace{-4pt}
\begin{align}
\gamma_{t} &= U'_{t}(d_{t}) = \pi_t^+ + \alpha_{t} = \pi_t^- - \beta_{t}, \label{kkt:1}
\\
\rho - \kappa &= c^g - \sum_{t\in\mathcal T}\mathbbm{E}[ \gamma_{t}\psi_{t}]\label{kkt:2}.
\end{align}
\vspace{-10pt}

From complementary slackness, if $d_{t}^->0$ then $\beta_{t}=0$. Analogously, if $d_{t}^+>0$ then $\alpha_{t}=0$. Thus,
\end{subequations}
\begin{subequations}
\vspace{-4pt}
\begin{align}
d_{t}^->0 &\implies \gamma_{t}=U'_{t}(d_{t})=\pi_t^- \label{kkt:3}\\
d_{t}^+>0 &\implies \gamma_{t}=U'_{t}(d_{t})=\pi_t^+ \label{kkt:4}
\end{align}
\vspace{-13pt}

\end{subequations}  Since $U'_{t}$ is strictly decreasing and invertible by assumption,
\eqref{eqn:cons_balance} and the complementary slackness conditions imply the following piecewise characterization, established in \cite{Alahmed&Tong:22TSG}.
\begin{lemma}[Optimal Operational Decisions]\label{lem:opt_decision}
\begin{subequations}\begin{align}
 &d_{t}^*
\hspace{-2pt}=\hspace{-2pt}
\begin{cases}
\bar d_t^+(\pi_t^+),
& \hspace{-4pt} g\psi_{t}\hspace{-2pt} < \hspace{-2pt}\bar d_t^+(\pi_t^+) \ (\text{import}),\\
\bar d_t^-(\pi_t^-),
& \hspace{-4pt} g\psi_{t} \hspace{-2pt}>\hspace{-2pt} \bar d_t^-(\pi_t^-) \ (\text{export}),\\
g\psi_{t},
& \hspace{-4pt} g\psi_{t}\hspace{-2pt}\in\hspace{-2pt}[ 
 \bar d_t^+(\pi_t^+), \bar d_t^-(\pi_t^-)]\ (\text{net-zero}).
\end{cases}\label{eq:net_demand_new}\\
 &d^{+*}_{t}
\hspace{-2pt}=\hspace{-2pt}
\begin{cases}
\bar d_t^+(\pi_t^+)-g\psi_{t},
&  g\psi_{t} < \bar d_t^+(\pi_t^+) ,\\
0,
& \text{otherwise}.
\end{cases}\label{eq:import_new}\\
 &d_{t}^{-*}
\hspace{-2pt}=\hspace{-2pt}
\begin{cases}
g\psi_{t} -\bar d_t^-(\pi_t^-),
&  g\psi_{t} > \bar d_t^-(\pi_t^-) ,\\
0,
& \text{otherwise},
\end{cases} \label{eq:export_new}
\end{align}\end{subequations}
\end{lemma}
\noindent where thresholds $\bar d_t^+(\pi_t^+)= \big(U'_{t}\big)^{-1}  (\pi_t^+)$ and $\bar d_t^-(\pi_t^-)= \big(U'_{t}\big)^{-1}  (\pi_t^-)$. The regimes in Equation \eqref{eq:net_demand_new}-\eqref{eq:export_new} are illustrated by Figure \ref{fig:opt_consump}. The blue curve shows the optimal consumption, and the green curve shows the optimal net demand given solar output $g\psi_{t}$. In each regime, generation $g\psi_{t}$ is compared to of the critical demand thresholds, 
$\bar d_t^-(\pi_t^-)$ and $\bar d_t^+(\pi_t^+)$. On the left side of Figure \ref{fig:opt_consump}, generation lies below $\bar d_t^+(\cdot)$, and the prosumer imports electricity to compensate for the difference between generation and demand at the import price. 
On the right side of Figure \ref{fig:opt_consump}, generation exceeds $\bar d_t^-(\cdot)$ and
the prosumer consumes to the point $\bar{d}^-_t(\cdot)$, after which additional utility gains are less valuable than the compensation gained by exporting the excess solar. On the left-hand side and right-hand side of the figure, as $g\psi_{t}$ increases, the net demand of the prosumer decreases as generation offsets demand. 
Finally, if local generation lies between the two thresholds, the prosumer is at a net-zero demand from the grid, at the level of PV generation. In this regime, the slope of the net demand curve is zero as the prosumer increases their consumption to exactly match generation.

\begin{figure}
    \centering
        \includegraphics[width=.9\linewidth]{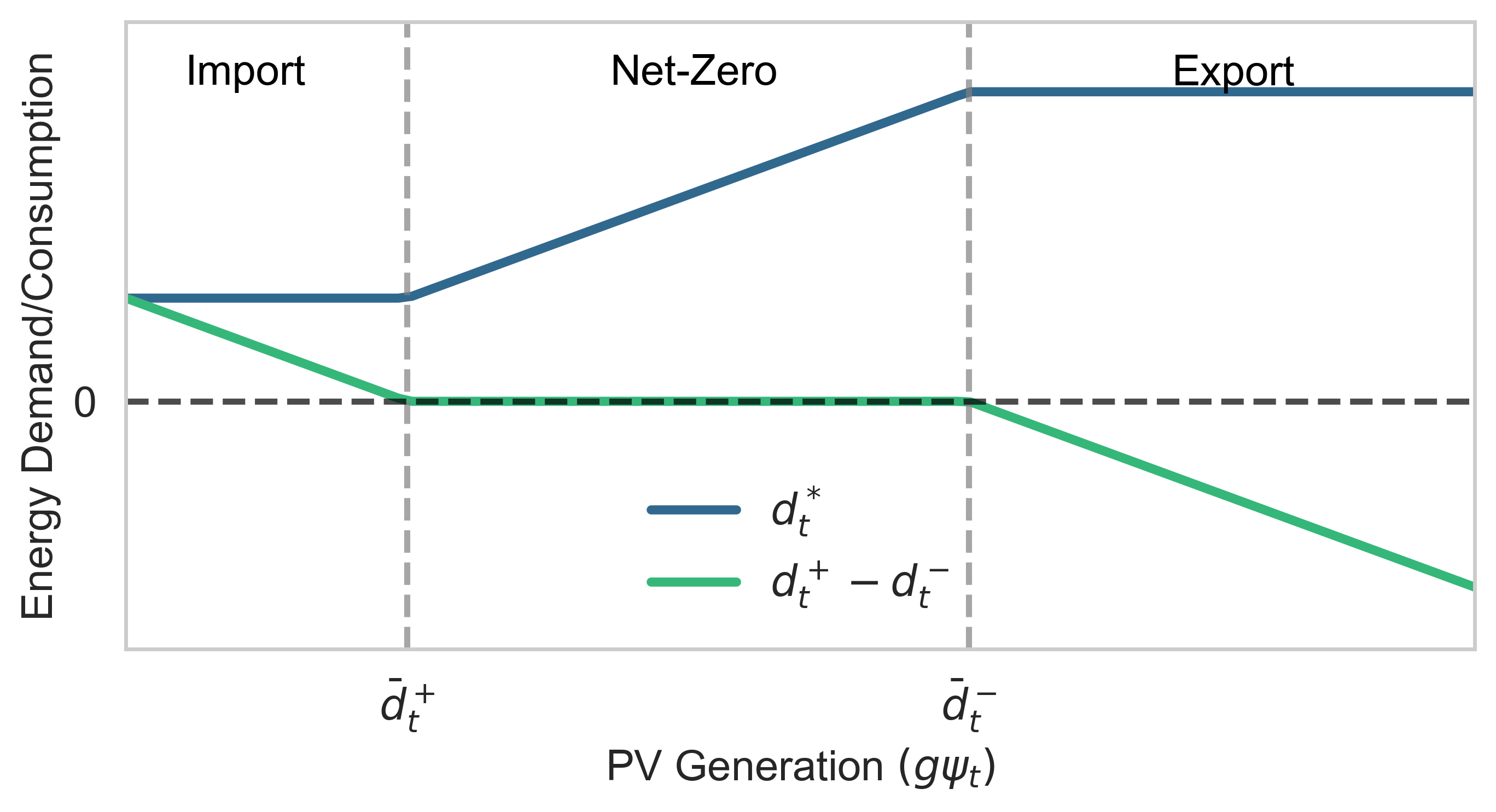} 
    \caption{Optimal consumption (blue) and optimal net demand (green) under NEM..}
    \label{fig:opt_consump}
\end{figure}
\vspace{3pt}

\noindent\textbf{\textit{{Optimal Prosumer Investment}}}
Using \eqref{kkt:1} and \eqref{kkt:2}, we find an equilibrium condition for interior solar investment; when $g^*\in(0,\bar{g})$, $\rho=0$ and $\kappa = 0$ due to complementary slackness. 
We can then characterize the marginal value of solar generation at the optimum and define the marginal value of PV for any level of investment: 

\begin{theorem}[Optimality Condition for Solar Investment]\label{thm:optimality_cond} The optimal PV capacity, $g^*$, is given by
\vspace{-4pt}
\begin{align}g^\ast = \max\{0, \min\{\bar{g}, g^\dagger\}\},\label{eq:optdagger}\end{align}\vspace{-16pt}

    where $g^\dagger$ is defined by the marginal value of interior solar investment determined by 
$\psi$-weighted average of the relevant period prices and marginal utilities:\vspace{-4pt}
\begin{subequations}
\begin{align}
&c^g 
=F(g^\dagger)\label{eq:value_solar_theta_exp}\\
&
F(g) := \mathbbm{E}\Big[
\sum_{t:g\psi_{t} < \bar{d}_t^+} \hspace{-10pt} \psi_{t}  \pi_t^+
 \hspace{-2pt}+  \hspace{-22pt}
\sum_{t:g\psi_{t} \in \big[ \bar{d}_t^+ , \bar{d}_t^-\big]}   \hspace{-23pt}\psi_{t} U'_{t}(\psi_{t}g)
 +  \hspace{-13pt}
\sum_{t:\psi_{t} g >  \bar{d}_t^-}\hspace{-10pt} \psi_{t}  \pi_t^-\Big].\label{eq:value_solar_exp}
\end{align} \end{subequations}
\end{theorem}
\vspace{-3pt}

Theorem \ref{thm:optimality_cond} states that, at optimality, the marginal value of solar is equal to the marginal cost if the optimal capacity is between the maximum and minimum PV capacities. If it is at the maximum (minimum), the marginal value can exceed (resp., be less than) the marginal cost. In \eqref{eq:value_solar_exp}, the marginal kW of solar capacity is valued by its capacity-factor-weighted contribution across periods: either displacing imports at price $\pi^+_t$, enabling additional consumption valued at marginal utility, or generating exports compensated at $\pi^-_t$.  As $g\psi_t$ increases from zero, period $t$ begins in the import regime, moves to the net-zero regime, and finally enters the export regime, with some periods incapable of reaching the net-zero or export regimes if $\bar \psi_{t}\bar g$ is smaller than $\bar d_t^+$ or $\bar d_t^-$, respectively. Because $\psi_t$ is a continuous random variable distributed in $[0,\bar\psi_{t}]$, any period with non-zero demand that can be in the export term has a non-zero probability associated with belonging to the import or net-zero terms. Figure \ref{fig:MarginalValue} shows a plot of the marginal value curve when there is only one period in $\mathcal{T}^g$. The flat section before $\underbar{d}^+$ corresponds to when no period is in the net-zero region. After this point, the marginal value curve in Figure \ref{fig:MarginalValue} is strictly decreasing, implying a unique solution when $g^*>\underbar d^+$, where $\underbar{d}^+ = \min_t \bar {d}_t^+$.
\begin{figure}
    \centering
    \includegraphics[width=\linewidth]{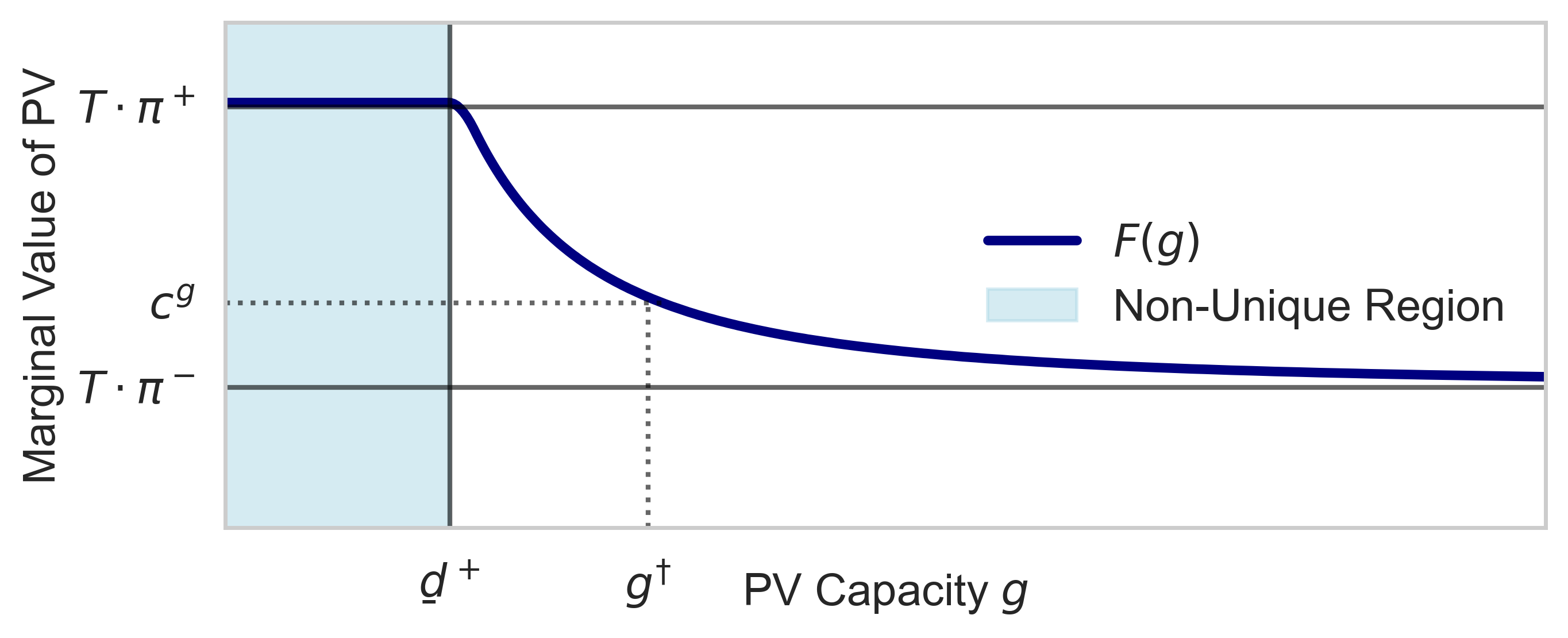}
    \caption{The marginal value of an additional unit of solar for a prosumer. }
    \label{fig:MarginalValue}
\end{figure}
The structure evident in Figure \ref{fig:MarginalValue}, a region of flat marginal value followed by a strictly decreasing region, describes the structure of marginal value curves beyond this illustrative example. Indeed, whenever we have a period for which there is a positive probability of the net-zero region, $F(g)$ is strictly decreasing in $g$. The implications of this are stated in the following corollary:  

\begin{corollary}\label{corr:optdaggernvert}
Equations \eqref{eq:value_solar_theta_exp}-\eqref{eq:value_solar_exp} have a unique solution whenever $c^g \neq F(0)$. If $c^g=F(0)$, the optimal solution is set valued, $g^*\in[0,\underbar{d}^+]$, and $F(g)$ is constant.
\end{corollary}
Theorem \ref{thm:optimality_cond} characterizes optimal PV investment in closed form, and Corollary \ref{corr:optdaggernvert} describes when optimal investment is unique or set-valued. These findings have key implications for the smoothness of investment changes in the sensitivity analysis.

\noindent\textbf{\textit{Sensitivity Analysis and Comparative Statics:}}
Using the characterization of optimal investment in Theorem \ref{thm:optimality_cond} and Corollary \ref{corr:optdaggernvert}, we perform a sensitivity analysis to better understand how changing prices cause shifts in optimal investment, optimal demand, and payments.
For interior solutions, the equality in \eqref{eq:value_solar_theta_exp} must hold even when parameters shift. Thus, the derivative of $F(g^*)-c^g$ with respect to prices and PV cost is equal to zero for interior solutions. We take the derivative, apply the chain rule, and rearrange terms to find:
 \vspace{-1pt}
\begin{align}
\frac{d g^*}{d  \pi_{\tau}^{\pm}}
\!=\!
\frac{
-\frac{\partial }{\partial \pi_{ \tau}^{\pm}}F(g^*)
}{
\frac{\partial }{\partial g}F(g^*)
},
\ 
\frac{d g^*}{d c^g}
\!=\!
\frac{
1
}{
\frac{\partial }{\partial g}F(g^*)
}.
\label{eq:gprime_pi_plusminus}
\end{align}
\vspace{-6pt}

Recall that $F(g^*)$ is decreasing in $g^*$.
Additionally, from \eqref{eq:value_solar_theta_exp}-\eqref{eq:value_solar_exp}, we can see that the $F(\cdot)$ is increasing in $\pi_\tau^\pm$ and must increase linearly with $c^g$. Using these partial derivatives, we can identify how much optimal PV capacity changes with prices and PV costs:
\begin{proposition}\label{prop:monotonicity}
Optimal investment is weakly increasing in $\pi^\pm_\tau$ and weakly decreasing in $c^g$. We further break down the movement of $g^*$ for each of the following cases:

\textbf{Case 1 (Capacity at bounds):} If $
c^g \not \in [F(\bar g), F(0)]
$, then $g^*\in\{0,\bar g\}$ and
$
d g^*/d \pi_\tau^+
=
d g^*/d \pi_\tau^-
=
d g^*/d c^g
=0.
$

\textbf{Case 2 (Interior capacity):} If $c^g \in (F(\bar g),F(0))$
and $g^*\not\in\{\bar d_t^+\}_{t\in\mathcal{T}^g}$. Then
\begin{subequations}
\begin{align}
\frac{d  g^*}{d  \pi_\tau^+}
&=
\frac{-\mathbbm E[\psi_\tau\cdot\mathbbm{1}\{\psi_\tau g^*<\bar d_\tau^+\}]}
{\mathbbm E[\sum_{t:\,\psi_t g^* \in (\bar d_t^+,\bar d_t^-)}
\psi_t^2\,U_t''(\psi_t g^*)]}
\;>\;0,
\label{eq:dg_dpi_plus_explicit}\\[6pt]
\frac{d  g^*}{d  \pi_\tau^-}
&=
\frac{-\mathbbm E[\psi_\tau\cdot\mathbbm{1}\{\psi_\tau g^*>\bar d_\tau^-\}]}
{\mathbbm E[\sum_{t:\,\psi_t g^* \in (\bar d_t^+,\bar d_t^-)}
\psi_t^2\,U_t''(\psi_t g^*)]}
\;>\;0,
\label{eq:dg_dpi_minus_explicit}\\[6pt]
\frac{d  g^*}{d  c^g}
&=
\frac{1}{
\mathbbm E[\sum_{t:\,\psi_t g^* \in (\bar d_t^+,\bar d_t^-)}
\psi_t^2\,U_t''(\psi_t g^*)]}
\;<\;0.
\label{eq:dg_solar_price_explicit}
\end{align}
\end{subequations}
When $g^*\in\{\bar d_t^+\}_{t\in\mathcal{T}^g}$, replace these with the corresponding left/right directional derivatives.

\textbf{Case 3 (Entry/Exit points):} If 
$
F(0)=c^g,
$
 the left-and right-hand derivatives are $0$ and \eqref{eq:dg_dpi_plus_explicit}-\eqref{eq:dg_dpi_minus_explicit} for $\pi_\tau^\pm$, and \eqref{eq:dg_solar_price_explicit} and $0$ for $c^g$. The directional derivatives are reversed when $
F(\bar g)=c^g
$.
\end{proposition}

Proposition \ref{prop:monotonicity} establishes that optimal solar capacity is weakly decreasing in PV costs and weakly increasing in both import and export prices. For interior solutions, the response of $g^*$ is smooth and governed by the marginal value condition in Theorem~\ref{thm:optimality_cond}: changes in prices or costs must be offset by adjustments in capacity so as to preserve the equality between the expected marginal value of PV and its marginal cost. 

The expressions in Case~2 make explicit how this adjustment depends on which operational regimes are active. Capacity changes are larger when periods have greater probability of being in the import or export regimes and smaller when periods have greater probability of being in the net-zero regime. Since marginal value in the net-zero region is given by $U'(\cdot)$ and is decreasing, additional capacity delivers declining value, which dampens the overall response of net demand to price changes. This causes the reduced sensitivity of net demand when generation is likely to fall in the net-zero regime.

The remaining cases complete the characterization of the capacity response to changes in parameters. When $c^g\notin[F(\bar g),F(0)]$, optimal capacity is at the bounds and small price changes have no effect on PV capacity. Around the entry point $F(0)=c^g$, the optimal capacity changes discretely, but the jump is limited to $\underline d^+$ and preserves monotonicity in all variables. Thus, Proposition~\ref{prop:monotonicity} implies that parameter changes induce at most one discrete adjustment in PV capacity, followed by a region of continuous response.

Changes in PV capacity have important implications for how operational decisions are altered by price changes. We demonstrate the effect of PV changes by decomposing the impact of an import price change on the optimal demand into direct and PV effects:
\begin{figure}
    \centering
    \includegraphics[width=\linewidth]{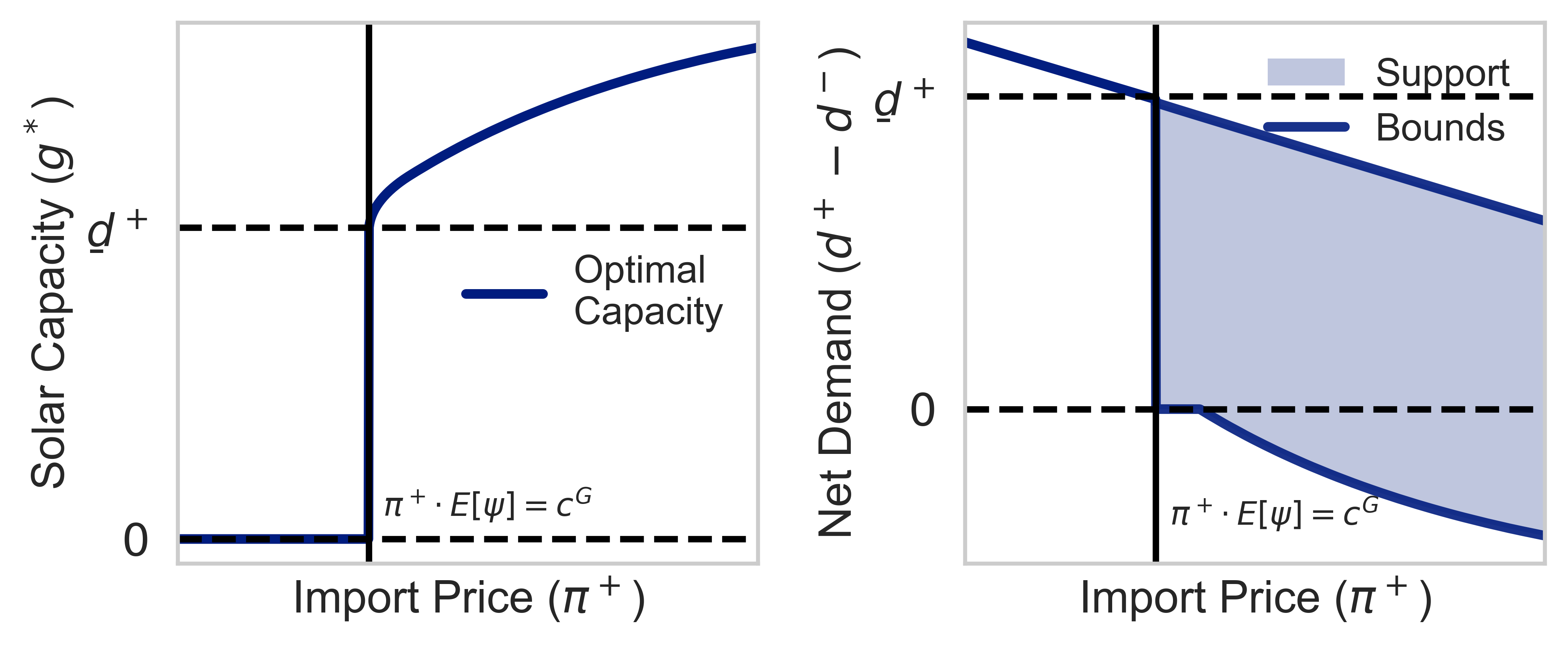}
    \caption{\textbf{Left:} Optimal solar as a function of the import price $\pi^+$. \textbf{Right:} Distribution of net demand for a single period as a function of the import price.}
    \label{fig:Awarepassive}
\end{figure}
\begin{align}
&\frac{d }{d \pi_\tau^{+}}\,
\mathbb E\!\left[d_{t}^{+*}-d_{t}^{-*}\right]\label{eq:nettrade_derivative_final}\\
&=
\underbrace{\mathbbm{1}\{t=\tau\}\cdot
\frac{\mathbb P\!\left(
\psi_{\tau} g^*
<
\bar d_{\tau}^{+}
\right)}{U''_{\tau}\!\big(\bar d_{\tau}^{+}\big)}
}_{\text{direct effect}}
\underbrace{
-
\mathbb E\Big[
\hspace{-12pt}
\sum_{t:\,\psi_{t}g^*<\bar d_{t}^+}\hspace{-15pt}\psi_{t}
\hspace{6pt}+\hspace{-10pt}
\sum_{t:\,\psi_{t}g^*>\bar d_{t}^-}\hspace{-15pt}\psi_{t}
\Big]
\!\cdot\!
\frac{\partial  g^*}{\partial \pi_\tau^{+}}
}_{\text{PV effect}}\!.
\nonumber
\end{align}
The decomposition in \eqref{eq:nettrade_derivative_final} highlights two distinct channels through which price changes affect net demand. The direct effect captures the within-period response: a change in $\pi_\tau^+$ alters $\bar d_t^+$ only for period $t=\tau$ and only when the prosumer is importing. The PV effect captures the effect of a change in optimal PV capacity adjustment: when prices change prosumers adjust solar investment, which in turn shifts net demand across all periods with positive generation. 

Figure~\ref{fig:Awarepassive} visualizes these effects using a one period example. The left panel shows how optimal solar investment responds to changes in $\pi^+$ including the discrete jump in capacity when $c^g=F(0)$, the a monotonic increase of PV capacity after this jump, and eventual saturation of PV capacity at $\bar g$ (not depicted). The right panel shows how these investment responses translate into net-demand outcomes. When PV capacity is zero, each $\pi^+$ corresponds to a single net-demand value. Once $g^*>0$, net demand spans a range due to variation in $\psi$. The downward slope of the upper bound reflects the direct effect of demand response to higher import prices, while the downward movement of the lower bound reflects the PV effect: as capacity increases, net demand continues to fall even in exporting states where the prosumer is not directly exposed to the import price.

Taken together, the analysis exposes four phenomena. First, price changes induce a direct, within-period response in net demand when the prosumer is exposed to the relevant price. Second, through the PV effect, price changes also alter net demand in other periods and regimes by shifting optimal solar capacity. Third, optimal PV investment exhibits a non-smooth entry threshold, followed by a continuous adjustment region and eventual saturation, which governs when the PV effect becomes active and how strongly it operates. Fourth, net-zero self-consumption periods dampen the investment response to price and cost changes, as additional generation is valued at declining marginal utility rather than market prices.

Endogenizing PV capacity also has important implications for payments, a central object for tariff design. For increases in the import or export price, treating PV capacity as fixed overestimates the resulting change in payment, as it ignores the capacity adjustment that lower net demand across generating periods. With endogenous PV, increases in capacity lower net demand, causing payment to rise less, or fall more, in response to import/export price increases than in the fixed-capacity case. Proposition~\ref{prop:payment} (Appendix \ref{app:prop2}) formalizes this decomposition into direct and PV effects. 

Finally, Table~\ref{tab:comparative-statics} summarizes the comparative statistics for investment, net demand, payment, and surplus. Changes in $c^g$ affect outcomes only through PV investment, whereas changes in $\pi_\tau^\pm$ combine direct and PV effects. Because $P^{\mbox{\tiny NEM}}$ and $S^*$ aggregate outcomes across periods, the reported comparative statistics should be interpreted as changes in total payment or surplus contributions summed over the periods in the relevant regime.

\begin{table}[!t]
\centering
\caption{Comparative statics analysis}
\label{tab:comparative-statics}
\begin{tabular}{@{}l*{9}{c}@{}}
\toprule
Variable & 
\multicolumn{3}{c}{$c^{g} \uparrow$} &
\multicolumn{3}{c}{$\pi_{\tau}^{+} \uparrow$} &
\multicolumn{3}{c}{$\pi_{\tau}^{-} \uparrow$} \\
\cmidrule(r){2-4} \cmidrule(lr){5-7} \cmidrule(l){8-10}
Region
& $+$ & $0$ & $-$ & $+$ & $0$ & $-$ & $+$ & $0$ & $-$ \\
\midrule
$d_{\tau}^{*}$ &
- & $\downarrow$ & $\downarrow$ &
$\downarrow$ & $\uparrow$ & - &
- & $\uparrow$ & $\downarrow$ \\
$d_{\tau}^{+*}-d_{\tau}^{-*}$ &
$\uparrow$ & - & $\uparrow$ &
$\downarrow$ & - & $\downarrow$ &
$\downarrow$ & - & $\downarrow$ \\
\midrule
$d_{t}^{*}$ &
& & &
- & $\uparrow$ & - &
- & $\uparrow$ & - \\
$d_{t}^{+*}-d_{t}^{-*}$ &
& & &
$\downarrow$ & - & $\downarrow$ &
$\downarrow$ & - & $\downarrow$ \\
\midrule
$P^{\mathrm{NEM}^*}$ &
$\uparrow$ & - & $\uparrow$ &
\(\text{X}\) & - & $\downarrow$ &
$\downarrow$ & - & $\downarrow$ \\
$S^{*}$ &
$\downarrow$ & $\downarrow$ & $\downarrow$ &
\(\text{X}\) & $\uparrow$ & $\uparrow$ &
$\uparrow$ & $\uparrow$ & $\uparrow$ \\
\bottomrule
\end{tabular}

\vspace{0.4em}

{\small
$\uparrow$ increase; $\downarrow$ decrease; - Unchanged; \(\text{X}\) Indeterminate.
}
\end{table}

\section{Case Study - A Household in Massachusetts}\label{sec:CaseStudy} 
We conduct a case study for a single family home in Massachusetts under a NEM tariff. To generate the parameters necessary for this case study, we use simulated hourly demand profiles and PV generation data for 2018 from the NREL ResStock dataset \cite{NREL:24}. We parameterize the utility function of every period $t$ (and the implied inverse-demand curves used in the surplus and payment calculations), using the same estimation approach as in the appendix of Alahmed and Tong \cite{Alahmed&Tong:22ACMSEIR}. To represent uncertainty in PV output within each settlement period, we modeled the capacity factor using a  normal random variable clipped to $[0,1]$: for each period $t$, we take the mean capacity factor from the profile in the ResStock simulated data and choose the standard deviation to match the month-specific standard deviation for capacity factors reported by Campbell et al.\cite{Campbell:24SD}. This did not exactly match the time scale of our monthly peak/off-peak periods, but provided a rough scale for the standard deviation of the capacity factors and was more appropriate than using a standard deviation derived for yearly, daily, or hourly data. Additional details on parameter choices are provided in Appendix \ref{app:empiric}.

We evaluated household operation and investment under five NEM tariffs. Each has different price settings based on summer/non-summer and peak/off-peak designations. Table~\ref{tab:price_parameters} summarizes the import and export price parameters used in the case study. For those with summer/non-summer differences, summer is defined as June-August. The \emph{Symmetric} prices are based upon current Massachusetts rates, and the compensation for exports is equal to that of imports. There is no peak/offpeak designation, and although this tariff has a peak period, it is the only period and captures all hours of the day. \emph{Proposed} prices are based upon a current proposal for time-of-use rates in Massachusetts. The billing periods are distinguished by peak/off-peak periods, with peak hours between 3pm and 8pm \cite{MARateTaskForce:24MDPU}. Similarly to \emph{Symmetric}, import prices under \emph{Proposed} are equal to export prices. \emph{Late} differs from \emph{Proposed} only in that the peak hours are shifted from 3-8pm to 6-11pm. We include \emph{Late} to illustrate the effect that the choice of peak hours has on solar profitability. Finally, \emph{Asymmetric} and \emph{Prop.-Asym.} tariffs are not based upon current prices or suggested prices in Massachusetts. We constructed them to illustrate how prosumer PV investment and operations for asymmetrical prices that are similar to current and proposed tariffs.

\begin{table}[t]
\centering
\caption{Import and export price parameters by tariff design}
\label{tab:price_parameters}
\resizebox{\columnwidth}{!}{%
\begin{tabular}{lcccccccc}
\toprule
 & \multicolumn{4}{c}{\textbf{Summer}} & \multicolumn{4}{c}{\textbf{Non-Summer}} \\
\cmidrule(lr){2-5} \cmidrule(lr){6-9}
\textbf{Tariff} 
& $\pi^{+}_{pk}$ & $\pi^{+}_{op}$ & $\pi^{-}_{pk}$ & $\pi^{-}_{op}$
& $\pi^{+}_{pk}$ & $\pi^{+}_{op}$ & $\pi^{-}_{pk}$ & $\pi^{-}_{op}$ \\
\midrule
Symmetric
& 0.35 & - & 0.35 & -
& 0.35 & - & 0.35 & - \\

Asymmetric
& 0.35 & 0.35 & 0.16 & 0.16
& 0.35 & 0.35 & 0.16 & 0.16 \\

Proposed
& 0.73 & 0.29 & 0.73 & 0.28
& 0.48 & 0.29 & 0.48 & 0.29 \\

Late
& 0.73 & 0.29 & 0.73 & 0.28
& 0.48 & 0.29 & 0.48 & 0.29 \\

Prop.-Asym.
& 0.55 & 0.45 & 0.25 & 0.20
& 0.35 & 0.24 & 0.15 & 0.10 \\

\bottomrule
\end{tabular}%
}
\end{table}

Figure~\ref{fig:marg_value_emp} illustrates the marginal value of an incremental increase in installed solar capacity, \eqref{eq:value_solar_exp}, to the prosumer in our case study for each tariff in Table \ref{tab:price_parameters}. 
\begin{figure}
    \centering
    \includegraphics[width=1\linewidth]{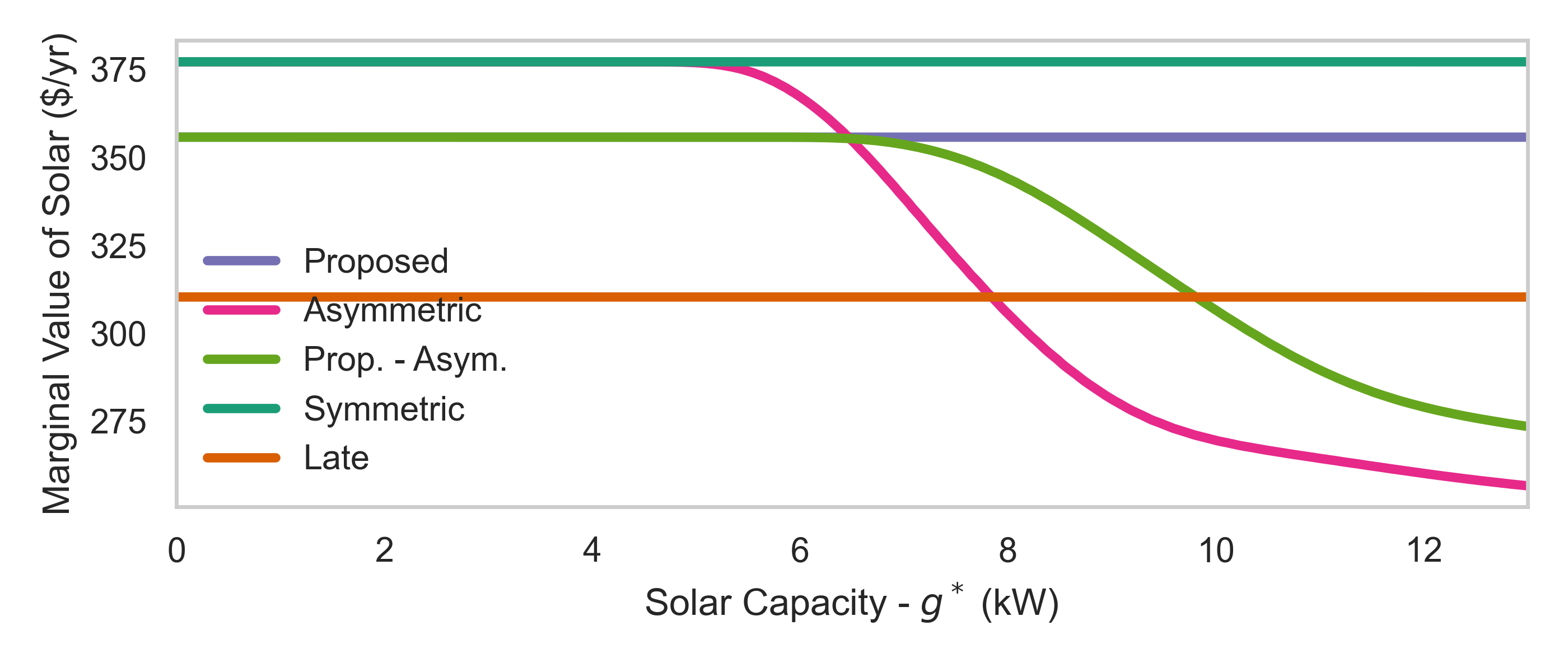}
    \caption{Marginal value of an additional unit of solar capacity. }
    \label{fig:marg_value_emp}
\end{figure} When prices are symmetric, marginal values are constant, and solar investments will be all-or-nothing. There is almost no net-zero region as $\bar d_t^+=\bar d_t^-+\epsilon$ with $\epsilon = 10^{-6}$. In contrast, when prices are asymmetric, the net-zero regime gives the marginal value function, $F(\cdot)$, curvature. At some point for both asymmetric tariffs, PV generation is sufficient to push some periods into the net-zero or export regimes. This decreases the profitability of larger installations under asymmetric tariffs relative to symmetric tariffs and smooths PV capacity changes so that optimal investment is no longer all-or-nothing. Instead, the jump is limited and elsewhere optimal capacity moves continuously with parameter changes (as discussed in Proposition \ref{prop:monotonicity}).

Another critical effect of tariffs on optimal investment is shown by the difference between \emph{Proposed} and \emph{Late}. Since the peak period for \emph{Late} occurs outside of high solar irradiance hours, PV generation displaces lower value electricity import and earns a lower payment when exporting than under \emph{Proposed}.

For \emph{Symmetric}, \emph{Proposed}, and \emph{Late}, Figure~\ref{fig:contour} provides contour plots of optimal installed solar capacity, total net demand with endogenized PV capacity, total net demand when PV capacity is fixed, total prosumer payment with endogenized PV capacity, and total prosumer payment when PV capacity is fixed. For fixed PV capacity, we use the optimal PV capacity at the original tariffs, \ie the tariffs without any changes to import or export prices. We also set $c^g = 342$ \$/kW and $\bar{g} = 13$ kW (Appendix \ref{app:empiric}). Import prices vary from baseline in all periods by 0 to 0.15 \$/kWh, and export prices vary from baseline in all periods by -0.15 to 0 \$/kWh (baselines are detailed in Table \ref{tab:price_parameters}). \begin{figure}
    \centering
    \includegraphics[width=1\linewidth]{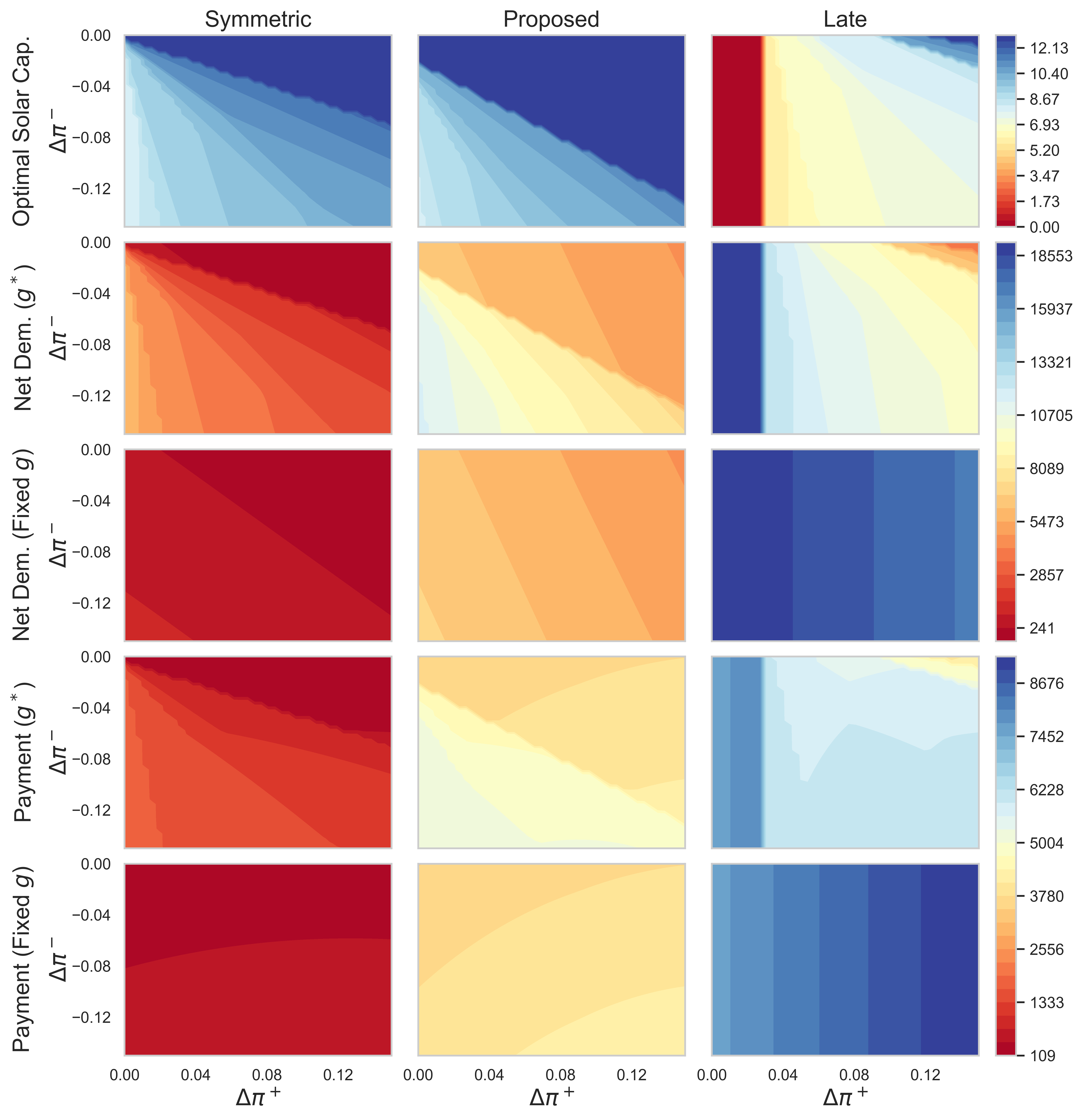}
    \caption{Contour plots comparing the impact of price changes on the prosumer's optimal solar capacity, total net demand with optimal PV capacity, total net demand with fixed PV capacity, total payment with optimal PV capacity, and total payment with fixed PV capacity.}
    \label{fig:contour}
\end{figure}

Figure~\ref{fig:contour} highlights several results shown in Section \ref{sec:Optimal}. \textbf{(i)} Optimal PV capacity is monotonically increasing in the import and export prices, as indicated by Prop. \ref{prop:monotonicity}. \textbf{(ii)} Net demand decreases in import and export prices due to an increase in generation and a decrease in consumption, reflecting the negative sign of the derivative in \eqref{eq:nettrade_derivative_final}. \textbf{(iii)} The combination of decreasing demand and increasing solar generation leads to a non-monotonic effect of import prices on payment, seen in Prop. \ref{prop:payment} and visible in all three optimal payment plots. The \emph{Late} plots show this the most clearly, where the transition from no PV to positive PV corresponds to a sharp decrease in payment in the optimal payment plots. Payment then increases for some $\pi^-$ perturbations but not all. While not shown here, prosumer surplus decreases (increases) monotonically with import (resp. export) price. Thus, there may be certain regions for which an increase in import price both lowers the payment of the prosumer and decreases their payment to the DU. \textbf{(iv)} Comparing the endogenous-PV rows to their fixed-PV counterparts (rows 2 vs.\ 3 and rows 4 vs.\ 5) shows that allowing PV capacity to adjust materially changes the implied sensitivity of both net demand and payment to tariff reforms. In particular, an increase in the import price produces a smaller increase in payment under endogenous PV than under fixed PV (and in some regions a larger decrease), whereas a decrease in the export price produces a larger increase in payment under endogenous PV than under fixed PV. Figure \ref{fig:contour} demonstrates that the importance of endogenous PV capacity to changes in net demand and payment discussed in Section \ref{sec:Optimal} are not just theoretical artifacts but apply to a realistic setting.

\section{Conclusion}\label{sec:conclusion}
This paper characterizes the joint investment and operational decisions of flexible prosumers with behind-the-meter solar generation under asymmetric, time-varying NEM tariffs. By endogenizing PV capacity, we derive conditions that link optimal investment to the marginal value of solar generation across three operational regimes: import, net-zero self-consumption, and export.

Asymmetric import and export prices give rise to an intermediate net-zero self-consumption region, which plays a central role in shaping investment incentives, which depend on how generation is distributed across these regimes. Optimal PV capacity increases with import and export prices, but its sensitivity to price changes is attenuated when a larger share of generation falls in the net-zero self-consumption regime. In this region, additional generation is valued at declining marginal utility rather than fixed market prices, which dampens investment responses.  

Endogenizing PV capacity reveals a PV effect: tariff changes in one period or regime affect net demand and consumption in all periods with positive generation through adjustments in optimal solar capacity. 
The PV effect also has direct implications for DU cost recovery. An increase in the import price raises payments per unit imported, but it may also induce additional PV investment that reduces net imports across generating periods, attenuating or even decreasing them. By contrast, decreases in the export price increase DU revenues through both the direct and PV effects. Ignoring investment responses therefore overstates the revenue impact of import price increases and understates the impact of an export price decrease.

A case study based on Massachusetts demand, PV generation, and tariffs illustrates our analytic findings under realistic conditions. The results show how tariff asymmetry and the timing of peak periods relative to solar production shape the marginal value of PV, smooth investment responses, and alter net demand and payments. Together, the analytical and numerical findings highlight the importance of accounting for endogenous PV investment when evaluating NEM tariffs.

This work opens several avenues for future research. Successor NEM tariffs increase incentives for storage adoption, motivating joint investment and operational analysis of PV–battery systems. In addition, relaxing the assumption of separable utility across periods would allow for further insight into prosumer behavior under NEM tariffs.

\section{AI Usage Disclosure}
During the preparation of this manuscript, the authors used ChatGPT for content assistance, specifically in drafting and condensing exposition. The authors reviewed and revised the material generated and assume full responsibility for the content of this publication.
\newpage

\appendix 
 \subsection{Statement and Proof of  Lemma \ref{lem:bireform}} \label{app:bireform }
\begin{lemma}\label{lem:bireform}
Suppose $\pi_t^+> \pi_t^-$ for all $t$. Then, for any solution to \eqref{eqn:prosumer_prob} that is optimal, we have
$
d_{t}^+ \cdot d_{t}^- = 0 \quad \forall\, t.
$
 \end{lemma}
Fix $t$ and consider any feasible point.
If at most one of $d_{t}^+,d_{t}^-$ is positive, the condition already holds. 
So assume $d_{t}^+>0$ and $d_{t}^->0$.

Pick any $\varepsilon\in(0,\min\{d_{t}^+,d_{t}^-\})$ and define
$
\tilde d_{t}^+ := d_{t}^+ - \varepsilon,
\tilde d_{t}^- := d_{t}^- - \varepsilon.
$
Then $\tilde d_{t}^+\ge 0$, $\tilde d_{t}^-\ge 0$, and
\begin{align*}
\tilde d_{t}^+ - \tilde d_{t}^- 
= (d_{t}^+ - \varepsilon) - (d_{t}^- - \varepsilon) 
= d_{t}^+ - d_{t}^- 
= d_{t} - \psi_{t}g,
\end{align*}
so the power-balance constraint \eqref{eqn:cons_balance} remains satisfied. 
Since we are only changing $d_{t}^+,d_{t}^-$ and keeping $d_{t}$ fixed, the electricity utility term $U_{t}(d_{t})$ is unchanged, 
as are $-c^g g$ and $\pi^c$. 
The only part of the objective in \eqref{eqn:prosumer_opt} that changes in period $t$ is
$
-\pi_t^+ d_{t}^+ + \pi_t^- d_{t}^-.
$
Under the modified pair,
\begin{align*}
-\pi_t^+ \tilde d_{t}^+ + \pi_t^- \tilde d_{t}^-
&= \big(-\pi_t^+ d_{t}^+ + \pi_t^- d_{t}^-\big) + \varepsilon(\pi_t^+ - \pi_t^-).
\end{align*}
Because $\pi_t^+\ge \pi_t^-$, we have $\varepsilon(\pi_t^+ - \pi_t^-)\ge 0$, 
with strict inequality if $\pi_t^+>\pi_t^-$. 
Thus the modified point is feasible and yields an objective value that is weakly larger 
(strictly larger when $\pi_t^+>\pi_t^-$) than the original one. 
Repeating this reduction step until one of $d_{t}^+,d_{t}^-$ hits zero produces a feasible point 
with weakly higher objective value and $d_{t}^+ d_{t}^-\!=\!0$. 
Hence no optimal solution can have $d_{t}^+\!>\!0$ and $d_{t}^-\!>\!0$.
\hfill $\blacksquare$

\textbf{Lagrangian}
The Lagrangian associated with problem \eqref{eqn:prosumer_prob} is
\begin{align}
\mathcal{L}
\!=& 
\mathbb E\Big[\!
\sum_{t\in\mathcal T}(
U_{t}(d_{t})
\!-\!\pi_t^+ d_{t}^+ + \pi_t^- d_{t}^-
)
\Big]\!
- \!c^g g\nonumber\\
&\!+ \!\mathbb E\Big[\!
\sum_{t\in\mathcal T} 
(-\alpha_{t}d_{t}^+ \!- \!\beta_{t}d_{t}^-
\!+\! 
\gamma_{t}(d_{t}^+ \!-\! d_{t}^- \!-\! d_{t} \!+\! \psi_{t} g))
\Big] \nonumber\\
&\!-\! \pi^c -\kappa g\! + \!\rho (g - \bar g)
\label{lag:new_final2}
\end{align}
where $\alpha_{t},\beta_{t},\kappa,\rho \ge 0$ and $\gamma_{t}\in\mathbb R$.
The expectation is taken with respect to the distribution of capacity factors $\psi_{t}$.

\textbf{Stationarity:}
The stationarity conditions are
\begin{subequations}
\begin{align}
\frac{\partial\mathcal L}{\partial d_{t}}
&\!=\! U'_{t}(d_{t})\!-\!\gamma_{t}\!=\!0,\label{stat:cons}\\ 
\frac{\partial\mathcal L}{\partial d_{t}^+}
\!&=\! -\pi_t^+\!-\!\alpha_{t}\!+\!\gamma_{t}\!=\!0,
\frac{\partial\mathcal L}{\partial d_{t}^-}
\!=\! \pi_t^-\!-\!\beta_{t}\!-\!\gamma_{t}=0,\label{stat:demand}\\
\frac{\partial\mathcal L}{\partial g}
&\!=\! -c^g\!-\!\kappa\!+\!\rho
 \!+\!\mathbb E\!\Big[\sum_{t\in\mathcal T}\gamma_{t}\psi_{t}\Big]=0.
\label{stat:minus_g}
\end{align}
\end{subequations}

\textbf{Primal Feasibility:} Primal feasibility entails,
\begin{align*}
d_{t}^+ \ge 0,\ 
d_{t}^- \ge 0,\ 
0 \le g \le \bar g,\ 
d_{t}^+ - d_{t}^- - d_{t} + \psi_{t} g = 0,
\end{align*}
holding almost surely.

\textbf{Dual Feasibility:} These are are
$
\alpha_{t},\ \beta_{t},\ \kappa,\ \rho \ge 0,
$.

\textbf{Complementary Slackness:}
\begin{align*}
\alpha_{t} d_{t}^+ = 0,\quad
\beta_{t} d_{t}^- = 0,\quad
\kappa g = 0,\quad
\rho( g - \bar g ) = 0.
\end{align*}
\vspace{-8pt} \subsection{Proof of Lemma \ref{lem:opt_decision}:} 
From \eqref{stat:cons} and \eqref{stat:demand},
\begin{align*}
\gamma_{t}
&= U'_{t}(d_{t})
= \pi_t^+\!+\!\alpha_{t}
= \pi_t^-\!-\!\beta_{t},
\\
\rho\!-\!\kappa
&= c^g-\mathbb E\!\Big[\sum_{t\in\mathcal T}\gamma_{t}\psi_{t}\Big].
\end{align*}
From complementary slackness, if $d_{t}^->0$ then $\beta_{t}=0$. Analogously, if $d_{t}^+>0$ then $\alpha_{t}=0$. Thus,
\begin{subequations}
\begin{align*}
d_{t}^-\!>\!0 &\!\implies\!\! \gamma_{t}\!=\!U'_{t}(d_{t})\!=\!\pi_t^-,\  
d_{t}^+\!>0\! \!\implies \!\!\gamma_{t}\!=\!U'_{t}(d_{t})=\pi_t^+. 
\end{align*}
Next use the balance constraint $d^+-d^-=d-g\psi$.
If $d>g\psi$, then $d-g\psi>0$ so we must have $d^+>0$ (and therefore $d^-=0$); thus
$d=\bar d^+(\pi^+)$, and substituting into the balance constraint gives
$d^{+}=\bar d^+(\pi^+)-g\psi$, $d^-=0$. Feasibility requires
$\bar d^+(\pi^+) > g\psi$, which is exactly the import condition.

If $d<g\psi$, then $d-g\psi<0$ so we must have $d^->0$ (and therefore $d^+=0$); thus
$d=\bar d^-(\pi^-)$ and $d^- = g\psi-\bar d^-(\pi^-)$, $d^+=0$. Feasibility then requires
$g\psi>\bar d^-(\pi^-)$.

Finally, if neither feasibility condition holds, \ie
$g\psi\in[\bar d^+(\pi^+),\,\bar d^-(\pi^-)]$, then the two strict cases above are
infeasible. The only feasible point satisfying the balance constraint with
$d^+\!=\!d^-\!=\!0$ is $d\!=\!g\psi$, which yields the net-zero regime. Substituting into
the balance constraint gives $d^+=d^-=0$.
Combining the three cases yields \eqref{eq:net_demand_new}-\eqref{eq:export_new}.
\end{subequations}
\hfill $\blacksquare$
\vspace{-8pt} \subsection{Proof of Theorem \ref{thm:optimality_cond}:}

 Any optimal solution satisfies the stationarity condition \eqref{kkt:2}. In particular, if $0<g<\bar g$ then
$\kappa=\rho=0$ and hence,
\vspace{-2pt}
\begin{align}
c^g
=
\mathbb E\Big[
\sum_{t\in\mathcal T}\gamma_{t}\psi_{t}
\Big], \label{eqn:cgequals}
\end{align}
\vspace{-10pt}

\noindent which is equivalent to \eqref{eq:value_solar_theta_exp}
Using $\gamma_t = U_t'(d_t)$ from \eqref{kkt:1} and complementary slackness from
\eqref{kkt:3}-\eqref{kkt:4}, $\gamma_t$ equals $\pi_t^+$ in the import regime,
$\pi_t^-$ in the export regime, and $U_t'(\psi_t g)$ in the net-zero regime.
Using Lemma~\ref{lem:opt_decision} to express these regimes as conditions on
$\psi_t g$ relative to $\bar d_t^+$ and $\bar d_t^-$, the right-hand side of
\eqref{eqn:cgequals} can be re-expressed as $F(g)$ in \eqref{eq:value_solar_exp}. Thus any
interior optimum $g^\dagger$ satisfies $c^g=F(g^\dagger)$.

Finally, since the feasible set for $g$ is the interval $[0,\bar g]$, if an
interior solution exists it must be $g^\dagger$, and otherwise the optimum is at
a boundary. Equivalently, the optimal capacity is the projection of $g^\dagger$
onto $[0,\bar g]$,
\[
g^*=\max\{0,\min\{\bar g,\,g^\dagger\}\},
\]
which is \eqref{eq:optdagger}.
\hfill $\blacksquare$

\vspace{-8pt} \subsection{Proof of Corollary \ref{corr:optdaggernvert}}
For any $g\in[0,\underline d^+]$ and any realization of $\psi_{t}\in[0,\bar{\psi}_{t}]$,
then $\psi_{t}g\le\bar d_t^+$ for all $t$, implying that all periods remain in the import
regime almost surely. Hence,
$
F(g)
=
F(0),
\quad
g\in[0,\underline d^+],
$
so $F(\cdot)$ is constant on $[0,\underline d^+]$ and therefore not invertible on this interval. For any $g>\underline d^+$, there exists at least one period $t$ such that
\[
\mathbb P\big(\bar d_t^-<\psi_{t}g<\bar d_t^+\big)>0.
\]
Since boundary events have probability zero, differentiation under the expectation yields
\begin{align}
\frac{\partial F}{\partial g}(g^*;\pi)
\!=\!
\mathbb E\!\big[
\sum_{t\in\mathcal T}\psi_{t}^2\,U''_{t}(\psi_{t}g)\,
\mathbbm 1\{\bar d_t^-\!<\!\psi_{t}g\!<\!\bar d_t^+\}
\big].
\label{eq:Fprime_continuous}
\end{align}
The contribution for every period is non-positive, and at least one is strictly negative with positive
probability. Thus,
\begin{align}
\frac{\partial F}{\partial g}(g^*;\pi)<0
\qquad
\forall g\in(\underline d^+,\bar g),
\label{eq:F_strict_decreasing}
\end{align}
so $F(\cdot)$ is strictly decreasing and therefore invertible on
$(\underline d^+,\bar g)$.
Consequently, the equation $F(g)=c^g$ has a unique interior solution
$g^\dagger$ whenever
$c^g\in\big(F(\bar g),\,F(\underline d^+)\big)$ and, by extension, whenever $c^g<F(0).
$ 
The flat region $[0,\underline d^+]$ is the only source of non-uniqueness; beyond this
region, continuous $\psi$ eliminates all discrete jumps in $F(\cdot)$.

\vspace{-8pt} \subsection{Proof of Proposition \ref{prop:monotonicity}}
Recall that $F(\cdot)$ is weakly decreasing in $g$ (strictly decreasing whenever the net-zero region has positive probability), while it is weakly increasing in each $\pi_\tau^\pm$.

\textit{Case 1 (Capacity at bounds).}
If $c^g\notin[F(\bar g),F(0)]$, then \eqref{eq:value_solar_theta_exp} has no interior solution. So $g^*\in\{0,\bar g\}$. For local perturbations of $(\pi_\tau^+,\pi_\tau^-)$ and $c^g$ that keep $c^g\notin[F(\bar g),F(0)]$, the argmax remains at the same boundary point, hence
$
d g^*/d\pi_\tau^+
=
d g^*/d\pi_\tau^-
=
d g^*/d c^g
=0.
$

\textit{Case 3 (Entry/Exit points).}
If $F(0)=c^g$, the boundary solution $g^*=0$ is optimal. For one-sided perturbations that move into the interior region (i.e., make $c^g\in(F(\bar g),F(0))$), the optimal $g^*$ becomes interior and its right-derivatives are given by the interior formulas in Case 2 (evaluated at the right limit), while the left-derivatives remain $0$ since $g^*$ stays pinned at $0$. Thus the left- and right-hand derivatives are $0$ and \eqref{eq:dg_dpi_plus_explicit}-\eqref{eq:dg_dpi_minus_explicit} for $\pi_\tau^\pm$, and \eqref{eq:dg_solar_price_explicit} and $0$ for $c^g$. The same argument applies at $F(\bar g)=c^g$, with left/right reversed since the boundary solution is then $g^*=\bar g$.

\textit{Case 2 (Interior capacity).}
Now suppose $c^g\in(F(\bar g),F(0))$ and $g^*\not\in\{\bar d_t^+\}_{t\in\mathcal T^g}$. Then $g^*\in(0,\bar g)$ and $F$ is differentiable at $g^*$, so differentiating $F(g^*)-c^g$ with respect to any parameter $\theta\in\{\pi_\tau^+,\pi_\tau^-,c^g\}$ yields
\begin{align}
0=\frac{\partial F}{\partial g}(g^*;\pi)\frac{d g^*}{d\theta}+\frac{\partial}{\partial \theta}\big(F(g^*;\pi)-c^g\big),
\label{eq:IFT_core_new}
\end{align}
and hence $d g^*/d\theta=-(\partial_\theta(F-c^g))/(\partial_g F)$.

Because $\psi_t$ is continuously distributed, regime boundary events have probability zero, so we may differentiate under the expectation. On any neighborhood with no regime switches,
\begin{align}
\frac{\partial F}{\partial g}(g^*;\pi)
=
\mathbbm E\!\Big[\sum_{t:\,\psi_t g^* \in (\bar d_t^+,\bar d_t^-)}
\psi_t^2\,U_t''(\psi_t g^*)\Big]
\;<\;0,
\label{eq:Fg_new}
\end{align}
since $U_t''<0$.

For $\pi_\tau^+$, only the import regime in period $\tau$ depends on $\pi_\tau^+$,
\begin{align}
\frac{\partial F}{\partial \pi_\tau^+}(g^*;\pi)
=
\mathbbm E\!\big[\psi_\tau\,\mathbbm 1\{\psi_\tau g^*<\bar d_\tau^+\}\big]
\;\ge\;0,
\label{eq:Fpi_plus_new}
\end{align}
and for $\pi_\tau^-$, only the export regime depends on $\pi_\tau^-$,
\begin{align}
\frac{\partial F}{\partial \pi_\tau^-}(g^*;\pi)
=
\mathbbm E\!\big[\psi_\tau\,\mathbbm 1\{\psi_\tau g^*>\bar d_\tau^-\}\big]
\;\ge\;0.
\label{eq:Fpi_minus_new}
\end{align}
Finally, $\partial(F-c^g)/\partial c^g=-1$. Substituting \eqref{eq:Fg_new}-\eqref{eq:Fpi_minus_new} into \eqref{eq:IFT_core_new} yields \eqref{eq:dg_dpi_plus_explicit}-\eqref{eq:dg_solar_price_explicit}. Since \eqref{eq:Fg_new} is negative and \eqref{eq:Fpi_plus_new}-\eqref{eq:Fpi_minus_new} are positive, it follows that $g^*$ is weakly increasing in $\pi_\tau^\pm$ and weakly decreasing in $c^g$.

If $g^*\in\{\bar d_t^+\}_{t\in\mathcal T^g}$, $F$ need not be differentiable; applying the same argument to left and right limits yields the corresponding directional derivatives.
\hfill $\blacksquare$

\vspace{-12pt} \subsection{Proposition \ref{prop:payment} and its Proof}
\label{app:prop2}
\begin{proposition}\label{prop:payment}
At points where $g^*$ is differentiable, the expected NEM payment varies with changes in the import or export price in period $\tau$ as
\begin{align}
\frac{\partial }{\partial \pi_\tau^{\pm}}\mathbbm{E}[P^{\mbox{\tiny NEM}*}]
\!=\!&
\underbrace{\mathbbm{E}\Big[\pm d_{\tau}^{\pm*}\Big]
\!+\!\pi_\tau^{\pm}\,
\frac{1}{U''_{\tau}\!\big(\bar d_{\tau}^{\pm}(\pi_\tau^{\pm})\big)}\,
\mathbb P\!\left(d_{\tau}^{\pm*}>0\right)}_{\text{direct effect}}
\nonumber\\
&\hspace{-3pt}
\underbrace{-
\mathbbm{E}\Big[\hspace{-18pt}
\sum_{t:\psi_{t}g^*<\bar d_{t}^+(\pi_t^+)}\hspace{-18pt} \pi_t^+\psi_{t}
+\hspace{-12pt}
\sum_{t:\psi_{t}g^*>\bar d_{t}^-(\pi_t^-)}\hspace{-18pt} \pi_t^-\psi_{t}
\Big]
\!\cdot\!
\frac{\partial g^*}{\partial \pi_\tau^{\pm}}}_{\text{PV effect}}.
\label{eq:dPstar_dpi_plus_cross_explicit}
\end{align}
\end{proposition}
\noindent\textit{Proof:}
Write the (expected) optimal payment as $\mathbbm{E}[P^{\mbox{\tiny NEM}*}(\pi)]
=\bar P(g^*(\pi),\pi)$, where $\bar P(g,\pi)$ is defined by \eqref{eq:Pi_NEM_def}
with $(d_t^{+*}(g,\pi),d_t^{-*}(g,\pi))$ given by Lemma \ref{lem:opt_decision}.
At any point where $g^*(\cdot)$ is differentiable, the chain rule gives
\begin{align}
\mathbbm{E}\Big[\frac{\partial}{\partial \pi_\tau^\pm}P^{\mbox{\tiny NEM}*}\Big]
=
\frac{\partial}{\partial \pi_\tau^\pm}\bar P(g^*,\pi)
+
\frac{\partial}{\partial g}\bar P(g^*,\pi)\cdot
\frac{\partial g^*}{\partial \pi_\tau^\pm}.
\label{eq:payment_chain_rule}
\end{align}

For the direct term, only the $\tau$-terms vary with
$\pi_\tau^\pm$ at fixed $g^*$. Using Lemma \ref{lem:opt_decision},
$d_{\tau}^{+*}(g^*,\pi)=[\bar d_\tau^+(\pi_\tau^+)-\psi_\tau g^*]^+$
and
$d_{\tau}^{-*}(g^*,\pi)=[\psi_\tau g^*-\bar d_\tau^-(\pi_\tau^-)]^+$.
Differentiating $\mathbb E[\pi_\tau^\pm d_\tau^{\pm*}]$ w.r.t.\ $\pi_\tau^\pm$ and using
$\frac{d}{d\pi_\tau^\pm}\bar d_\tau^\pm(\pi_\tau^\pm)=1/U_\tau''(\bar d_\tau^\pm(\pi_\tau^\pm))$
yields the direct-effect term in \eqref{eq:dPstar_dpi_plus_cross_explicit}.

For the PV term, differentiate $\bar P(g,\pi)$ w.r.t.\ $g$ using Lemma \ref{lem:opt_decision}:
$\partial d_t^{+*}/\partial g=-\psi_t\,\mathbbm 1\{\psi_t g<\bar d_t^+(\pi_t^+)\}$ and
$\partial d_t^{-*}/\partial g=\psi_t\,\mathbbm 1\{\psi_t g>\bar d_t^-(\pi_t^-)\}$ a.s.,
so
\begin{align}
\frac{\partial}{\partial g}\bar P(g^*,\pi)
=
-\mathbbm{E}\Big[\hspace{-18pt}
\sum_{t:\psi_{t}g^*<\bar d_{t}^+(\pi_t^+)}\hspace{-18pt} \pi_t^+\psi_{t}
+\hspace{-12pt}
\sum_{t:\psi_{t}g^*>\bar d_{t}^-(\pi_t^-)}\hspace{-18pt} \pi_t^-\psi_{t}
\Big].
\label{eq:dPdg_compact}
\end{align}
Substituting \eqref{eq:dPdg_compact} into \eqref{eq:payment_chain_rule} gives
\eqref{eq:dPstar_dpi_plus_cross_explicit}.
\hfill $\blacksquare$

The direct effect is exactly the effect of an import or export price change on the payment from import or export in period $t$ given a fixed $g$. The second effect, the PV effect, captures the indirect effect of a price change on imports and exports in all periods through the effect on optimal solar investment. One can see that there is an asymmetry in the signs of of effect for the two price changes. Even when the direct effect is positive, an increase in price can potentially decrease payments through its impact on optimal investment. Fourth, because the PV effect is always negative, payments may be more negatively responsive to price increases relative to fixed-capacity benchmarks. In contrast, an increase in the export price unambiguously lowers payment. Endogenized PV capacity causes net demand to decrease in in both the export and import prices, altering contributions to payment from periods not subject to the relevant prices. Similar effect occur for prosumer surplus in the net import regime (see Table \ref{tab:comparative-statics}).

\vspace{-8pt} \subsection{Justification of Table \ref{tab:comparative-statics}.}
All comparative statics are local and evaluated at points where regimes do not switch; boundary events have probability zero under the assumed continuity of $\psi_t$. We use on Lemma~\ref{lem:opt_decision} for the piecewise form of $(d_t^*,d_t^{+*}-d_t^{-*})$ and on Proposition~\ref{prop:monotonicity} for the signs of $dg^*/d\pi_\tau^\pm$ and $dg^*/dc^g$.

\paragraph{Own-period quantities (first row)}
For period $\tau$, $d_\tau^*=\bar d_\tau^+(\pi_\tau^+)$ in the import region, $d_\tau^*=g^*\psi_\tau$ in the net-zero region, and $d_\tau^*=\bar d_\tau^-(\pi_\tau^-)$ in the export region. Hence, under $c^g\uparrow$, only the net-zero case responds locally and decreases since $dg^*/dc^g<0$. Under $\pi_\tau^+\uparrow$, $d_\tau^*$ decreases in import (by $\frac{d}{d\pi}\bar d_\tau^+<0$), increases in net-zero (since $dg^*/d\pi_\tau^+>0$), and is unchanged in export. Under $\pi_\tau^-\uparrow$, $d_\tau^*$ is unchanged in import, increases in net-zero, and decreases in export (by $\frac{d}{d\pi}\bar d_\tau^-<0$). These give the  of the table.

For net demand, $d_\tau^{+*}-d_\tau^{-*}=\bar d_\tau^+(\pi_\tau^+)-g^*\psi_\tau$ in import, $0$ in net-zero, and $\bar d_\tau^-(\pi_\tau^-)-g^*\psi_\tau$ in export. Thus $c^g\uparrow$ raises net demand in both import and export, while leaving the net-zero case unchanged. An increase in either $\pi_\tau^+$ or $\pi_\tau^-$ lowers net demand in both import and export (directly through $\bar d_\tau^\pm$ when applicable and indirectly through $g^*$)..

\paragraph{Other period quantities ($t\neq\tau$)}
When only $(\pi_\tau^+,\pi_\tau^-)$ changes, $\bar d_t^\pm(\pi_t^\pm)$ is unaffected for $t\neq\tau$; the only channel is $g^*$. Hence $d_t^*=g^*\psi_t$ increases in the net-zero region and is locally unchanged in import and export, while $d_t^{+*}-d_t^{-*}$ decreases in both import and export and remains zero in net-zero. This yields the middle block of the table.

\paragraph{Payments}
Payment responses follow from Proposition~\ref{prop:payment} by combining a regime-specific direct effect in period $\tau$ with a PV effect operating through $dg^*/d(\cdot)$.  
If period $\tau$ is net-zero, neither channel is active locally and payment is unchanged.  
If period $\tau$ is importing, an increase in $\pi_\tau^+$ raises the per-unit import charge but also induces additional PV investment that reduces net imports across generating periods, rendering the net payment effect indeterminate.  
If period $\tau$ is exporting, $\pi_\tau^+$ affects payment only through the PV channel, yielding a decrease.  
By contrast, increases in $\pi_\tau^-$ reduce expected payment whenever period $\tau$ trades, since both the direct and PV effects act in the same direction.  
Finally, when $c^g$ increases, only the PV channel operates via $dg^*/dc^g<0$, and the resulting payment changes follow directly from the sign of $\partial P^{\mathrm{NEM}}/\partial g$, as summarized in Table~\ref{tab:comparative-statics}.

\paragraph{Surplus}
Surplus results follow from the envelope theorem. Since $\partial S^*/\partial c^g=-g^*\le 0$, surplus weakly decreases when $c^g\uparrow$ in all regimes.  
For $\pi_\tau^+$, surplus combines a within-period effect (when $\tau$ is importing) with an investment response, making the net effect indeterminate in the import regime and weakly positive otherwise.  
For $\pi_\tau^-$, higher export compensation and the induced increase in $g^*$ both raise surplus, yielding $S^*\uparrow$ across regimes.

\vspace{-8pt} \subsection{Case Study}

\label{app:empiric}
\noindent\textit{PV Parameters:} We use $\bar{g}=13$ kW as this is the largest solar size seen in the ResStock dataset for single-family homes in Massachusetts. This is an admittedly high level of PV for a residential home in MA and results in lower net demand than we would see for a smaller upper bound on investment. We use the following calculation to derive the yearly amortized cost of a kW of PV:
\[
c^g = 0.7\cdot 12 \cdot C \frac{r}{1 - (1+r)^{-n}}
\]
where
$C=375$ \$/kW,
$r = 0.055/12$, and $n = 120$
corresponding to a $10$ year, $5.5\%$ loan compounded monthly. This is the median loan recorded by MASSCEC: \url{https://www.masscec.com/program/mass-solar-loan}. We multiple by 12 to convert into a yearly cost and 0.7 to represent the federal tax credit. We do not include state level incentives.

\noindent\textit{Utility functions and Demand Parameterization:} We use the same methodology as in the appendix of \cite{Alahmed&Tong:22ACMSEIR}, using quadratic utility functions that correspond to linear demand curves. We assume a price elasticity of $-0.25$ as this lies within the estimates found in \cite{Asadinejad:18EPSR}. We use the demand within ResStock data and a electricity price of $0.35$ \$/kW to calculate the utility function parameters as in \cite{Alahmed&Tong:22ACMSEIR}.

\bibliographystyle{ieeetr}
\bibliography{references}
\balance

\endgroup

\end{document}